\documentclass[11pt,a4paper]{article}

\usepackage[utf8]{inputenc}
\usepackage[margin=1in]{geometry}
\usepackage{amsmath,amssymb,amsfonts,amsthm}
\usepackage{graphicx}
\usepackage{booktabs}
\usepackage{hyperref}
\usepackage{cite}
\usepackage{tikz}
\usetikzlibrary{shapes.geometric, arrows, positioning}

\hypersetup{
    colorlinks=true,
    linkcolor=blue,
    citecolor=blue,
    urlcolor=blue
}

\theoremstyle{definition}
\newtheorem{definition}{Definition}
\newtheorem{proposition}{Proposition}
\newtheorem{requirement}{Requirement}

\title{\textbf{Themis Consensus Extension:}\\
MEV Mitigation by Randomized Delayed Execution and Intent-Hiding Transactions in Application-Specific Blockchains}

\author{
  \textbf{Shoeb Siddiqui}, \textbf{Mateusz Nowakowski}, \textbf{Stanislav Vozarik}, \\ \textbf{Gleb Urvanov}, \textbf{Peter Kris} \\[0.5em]
}
\date{July 2026}

\begin{document}

\maketitle

\begin{abstract}
Maximal extractable value (MEV) arises when privileged participants select, exclude, insert, or reorder pending transactions for private gain. We specify and analyze the Themis Consensus Extension v1, first published by Mangata in 2021. The design separates value extraction by reordering (VER) from value extraction by denial (VED). For VER, block construction and execution occur across consecutive producers: one producer commits a transaction set, and the next derives a publicly verifiable, deterministic, previously unpredictable seed and executes a seed-determined, dependency-preserving permutation. For selective VED, a user may encrypt a transaction for a designated builder and executor. The builder removes an outer layer and commits the opaque inner ciphertext; the executor reveals and executes the plaintext only after commitment.

Under selfish but non-colluding validators, an adversary below the underlying consensus fault threshold, secure cryptography, and accountable role performance, the construction limits unilateral post-commit ordering control and hides transaction intent from relays and the builder. It does not provide send-order or receive-order fairness, complete censorship resistance, resistance to builder-executor collusion, or per-transaction price guarantees. We analyze probabilistic extraction, spam, dependent transactions, decryption liveness, session boundaries, total denial, and threshold coalitions. We also document the initial Aura-based Substrate implementation and its subsequent transition to a BABE-based sr25519/VRF seed path, together with delayed execution, Fisher-Yates shuffling, and Xoshiro256++. The result preserves the original proposal while narrowing its claims to explicit assumptions.
\end{abstract}


\clearpage

\section{Introduction}
Block producers commonly determine both which pending transactions enter a block and the order in which those transactions execute. In stateful financial applications, these powers can be monetized through frontrunning, backrunning, sandwiching, selective censorship, replacement, and related strategies. The term maximal extractable value (MEV) captures value obtained through privileged control over transaction inclusion, exclusion, insertion, or ordering. Early empirical work documented decentralized-exchange arbitrage bots, priority gas auctions, and resulting risks to consensus stability \cite{daian2020flash}.

The original Themis proposal described this asymmetry using an analogy to insider trading: an intermediary observes pending user intent, controls a consequential execution decision, and may trade against the user \cite{mangata2021themis}. The analogy is motivational rather than a legal characterization. The technical concern is that a protocol actor has powers unavailable to an ordinary user.

MEV responses can be organized around two broad strategies. Redistribution mechanisms expose or auction ordering opportunities and attempt to distribute their proceeds. Flashbots' early program exemplified this direction through transparent extraction infrastructure, auctions, and proposed benefit distribution \cite{obadia2020frontrunning}. Minimization mechanisms instead attempt to remove or constrain the powers that create an extraction opportunity. Themis belongs to the latter family.

This paper formalizes the design originally published as "Introducing Themis Consensus Extension v1" \cite{mangata2021themis}. We refer to it as TCE-v1 to distinguish it from a separate Byzantine consensus protocol later published under the name Themis \cite{kelkar2023themis}. The name in the original proposal refers to the Greek figure associated with justice.

The design targets application-specific blockchains, especially decentralized exchanges (DEXs), where consensus, fee rules, transaction semantics, wallets, and execution can be changed together. It has two components:
\begin{enumerate}
    \item \textbf{Randomized delayed execution for VER.} A block producer commits the transaction set for block $n$, but the next producer determines a verifiable random seed and executes that set during block $n+1$. Thus, the participant choosing content does not unilaterally choose order, and the participant deriving order does not unilaterally choose content.
    \item \textbf{Optional intent hiding for selective VED.} A user encrypts a transaction for a designated builder and executor. The builder can prepare an opaque inner ciphertext for inclusion but cannot identify its economic intent. The executor reveals the transaction only after inclusion.
\end{enumerate}

The original article presented the design as a reasonably complete MEV solution using conventional cryptography and suggested adaptation to both proof-of-stake (PoS) and proof-of-work (PoW) networks. Academic treatment requires narrower language. The current design is most natural when future producers are known, assumes non-collusion, and cannot prevent an authorized builder from denying all or arbitrary opaque transactions. We therefore describe it as an MEV-minimization mechanism under an explicit threat model, not as a universal elimination of MEV.

\subsection{Contributions}
This paper makes four contributions:
\begin{itemize}
\item It restates the original VER/VED decomposition as a precise protocol objective and defines the system, cryptographic, consensus, and economic assumptions needed by TCE-v1.
\item It specifies the delayed-execution pipeline, seed requirements, dependency-preserving shuffle, and two-layer intent-hiding workflow in a form suitable for implementation review.
\item It analyzes the guarantees and non-guarantees of v1, including probabilistic extraction, Sybil accounts, spam costs, decryption accountability, total denial, collusion, and consensus-threshold failure.
\item It records the initial Aura implementation and the subsequent transition to BABE, including the different randomness and seed paths documented across the project materials.
\end{itemize}

No new benchmark or production measurement is claimed. Statements about implementation status and chronology are historical reports from the original project team and public materials; reproducibility still requires immutable commits or releases for the specific Aura and BABE revisions discussed.

\section{Scope, Terminology, and Design Position}

\subsection{Conventional and unconventional MEV}
The original proposal adopts the early Flashbots distinction between conventional and unconventional MEV \cite{flashbots2020taxonomy, obadia2020frontrunning}. Conventional MEV consists of protocol-intended compensation such as transaction fees and block rewards. Unconventional MEV consists of gains from transaction reordering, insertion, exclusion, or censorship. This paper uses MEV to mean the unconventional component unless stated otherwise.

This scope choice is normative. It treats fees and issuance as necessary protocol incentives rather than harms to remove. It also avoids classifying MEV as intrinsically "good" or "bad." The relevant question is whether a privileged actor can appropriate an opportunity that ordinary users cannot access on comparable terms. Arbitrage may improve prices, for example, yet its allocation is asymmetric when only a builder can censor a user's arbitrage and substitute its own.

\subsection{Miner, validator, producer, and user}
The 2021 article uses \emph{miner} for both PoW miners and PoS validators. This paper uses \emph{validator} or \emph{producer} except when discussing PoW adaptation. A \emph{user} is modeled as a light client or wallet participant. A validator may also transact as a user, but the roles are separated analytically.

The term \emph{builder} denotes the producer that selects a transaction set. The term \emph{executor} denotes the next producer that determines the order and applies the state transition. These are temporal roles. A single validator may occupy both roles in different slots, and repeated selection can place the same validator in adjacent roles.

\subsection{Random ordering is not arrival-order fairness}
TCE-v1 does not attempt to prove that a transaction received or sent first will execute first. It instead attempts to make execution order unpredictable before the transaction set is committed and symmetric with respect to eligible transactions after commitment. This is materially different from receive-order, send-order, or relative order-fairness studied by Aequitas and Wendy \cite{kelkar2020order, kursawe2020wendy}.

\subsection{Application-specific deployment}
The design assumes that the chain can coordinate consensus, transaction fees, failure rules, slashing, wallet encryption, and application semantics. This is especially relevant to a DEX-specific chain with fixed fees and no priority fee auction. A different application may expose different MEV opportunities. Every deployment therefore needs an application-level MEV analysis rather than treating the consensus extension as sufficient by itself.

The original article described Themis as permissionless and based only on traditional cryptography. The cryptographic ingredients are conventional, but permissionlessness is not a standalone security property. It depends on the validator-selection mechanism, role discovery, key publication, transaction access, and the absence of gatekeeping in the protected submission path.

\section{System and Threat Model}

\subsection{Entities and slot pipeline}
Time is divided into numbered production slots $n=1,2,\dots$. Let $P_n$ denote the producer authorized for slot $n$. In the basic pipeline, $P_n$ serves two roles:
\begin{itemize}
\item as builder $B_n$, it selects and commits a transaction set $\mathcal{T}_n$ for block $n$;
\item as executor $E_{n-1}$, it derives the execution seed for the transaction set committed in block $n-1$, computes its execution order, and applies it to the prior state.
\end{itemize}

A public mempool or relay layer carries ordinary transactions. An optional targeted path carries encrypted transactions addressed to a known builder-executor pair. Consensus validates blocks, producer eligibility, seed proofs, execution results, and any evidence needed for slashing.

Each transaction $t$ has a sender account $a(t)$ and, when applicable, an account nonce $q(t)$. Transactions from one account may depend on previous transactions from that account. The protocol must therefore preserve valid nonce order while randomizing the interleaving between accounts.

\subsection{Consensus assumptions}
The safety and liveness of TCE-v1 inherit the assumptions of the underlying consensus protocol. The design assumes:
\begin{enumerate}
    \item the adversarial coalition remains below the consensus protocol's Byzantine or economic fault threshold;
    \item the authorized producer for each relevant slot can be identified and its public key is available;
    \item invalid seed proofs, invalid execution, or required-message omissions can be rejected or punished according to consensus rules;
    \item finality or fork-choice behavior is sufficient to define the committed transaction set before its execution order is finalized;
    \item the chain can represent the one-block execution delay without ambiguity in state roots, receipts, events, transaction status, and client queries.
\end{enumerate}

No consensus extension can reliably punish a coalition that already controls enough power to rewrite the ledger, suppress evidence, or redefine validity. The original article described this as the inability to stop a malicious majority. More generally, the limit is the fault threshold of the underlying consensus and governance system, which need not equal a simple numerical majority.

\subsection{Adversary model}
Version 1 primarily models a selfish, strategically acting validator that does not cooperate with the adjacent producer. The adversary may:
\begin{itemize}
\item inspect any plaintext information available to its role;
\item choose which transactions to include, subject to any protocol-level inclusion rules;
\item submit its own transactions and use multiple accounts;
\item withhold a block, seed, decryption, or execution result if the resulting penalty is acceptable;
\item exploit economic opportunities created by application state;
\item operate or influence relays, RPC endpoints, and search infrastructure.
\end{itemize}

The v1 security argument does not cover collusion between the builder that commits $\mathcal{T}_n$ and the executor that derives its order or reveals its encrypted contents. It also does not cover compromised wallets, broken cryptography, consensus-threshold control, governance capture, or side channels that reveal intent despite encryption.

\subsection{Cryptographic assumptions}
The order-seed mechanism requires a primitive with a deterministic or unique output for a given key and message, public verifiability, and unpredictability to parties lacking the secret key. A verifiable random function (VRF) is a natural realization \cite{micali1999vrf}. A deterministic, non-malleable unique-signature construction can serve a similar role if its exact security properties are established.

The intent-hiding path requires authenticated hybrid public-key encryption. Ciphertexts must be bound to the chain, session, target roles, transaction metadata, fee or resource commitment, expiry, and replay domain. Encryption must provide confidentiality and integrity, not merely obscurity. Key rotation and role changes must be aligned with ciphertext expiry.

\subsection{Economic assumptions}
The original design is coupled to a DEX with fixed transaction fees, no priority gas auction, exchange commissions, and charges for failed or canceled transactions. These rules make repeated extraction attempts costly. They do not cryptographically prevent spam or Sybil accounts. Economic claims therefore depend on fee levels, opportunity size, liquidity, account costs, and credible penalties.

\section{MEV Powers: Reordering and Denial}
The original proposal groups its targeted extraction surface into two operational powers.

\begin{definition}[Value extraction by reordering]
VER is value obtained by a participant that can determine or influence the relative execution order of transactions, including the insertion of its own transactions before or after a user's transaction.
\end{definition}

The label includes straightforward front- and back-running as well as order-sensitive interactions among contracts or application calls. It is deliberately broader than any single trading pattern.

\begin{definition}[Value extraction by denial]
VED is value obtained by a participant that can reject, delay, suppress, or replace a user's transaction. The participant may be a centralized relay, an RPC service, a mempool gateway, or a block producer.
\end{definition}

Selective denial is especially powerful when the actor sees a pure-gain opportunity such as arbitrage, censors the user's attempt, and submits a substitute transaction. The original article used a centralized Ethereum access provider as an example of a relay with such power. The example illustrates architectural concentration and does not allege specific misconduct.

This two-power decomposition is a design lens rather than a proof that every MEV strategy has a unique classification. Some attacks combine both powers. A sandwich, for example, may require insertion and ordering, while private order flow may also affect inclusion.

\subsection{Three levels of denial mitigation}
The original proposal distinguishes three levels, in increasing strength:
\begin{enumerate}
    \item validators \emph{should not} deny transactions, relying on policy or incentives;
    \item validators \emph{do not know which} transaction to deny, because intent is hidden;
    \item validators \emph{cannot} deny transactions, because inclusion is enforced independently of their discretion.
\end{enumerate}

The v1 encryption mechanism targets the second level. It reduces selective censorship based on transaction purpose. It does not achieve the third level because a builder can still omit all ciphertexts, a random subset, or every transaction except its own.

\section{Mitigating Reordering with Delayed Randomized Execution}

\subsection{Protocol overview}
Standard block production combines content selection and execution-order selection. TCE-v1 separates them in time. Transactions included in block $n$ are executed while block $n+1$ is being produced or imported. The producer of block $n+1$ derives a seed that was not available to the producer of block $n$ when the transaction set was committed.

Figure \ref{fig:pipeline} shows the pipeline. The one-slot delay creates what the original article called a "two-block HEAD": the latest block contains a committed but not yet executed transaction set, while the preceding set has obtained its final execution order and state transition. The article described this as doubling block-execution time. More precisely, the pipeline adds one block interval to application-visible completion while allowing execution and construction to overlap.

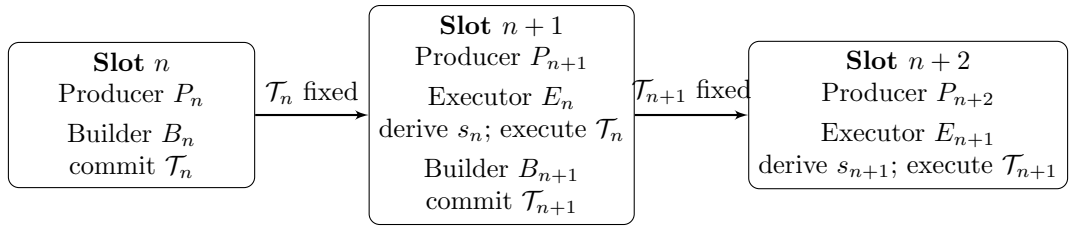
\begin{figure}[htbp]
\centering
\begin{tikzpicture}[node distance=1.8cm, auto, >=latex', every node/.style={font=\small}]
    \node [draw, rectangle, rounded corners, text width=3cm, align=center] (slot1) {\textbf{Slot $n$}\\Producer $P_n$\\[0.3em]Builder $B_n$\\commit $\mathcal{T}_n$};
    \node [draw, rectangle, rounded corners, right=1.5cm of slot1, align=center] (slot2) {\textbf{Slot $n+1$}\\Producer $P_{n+1}$\\[0.3em]Executor $E_n$\\derive $s_n$; execute $\mathcal{T}_n$\\[0.3em]Builder $B_{n+1}$\\commit $\mathcal{T}_{n+1}$};
    \node [draw, rectangle, rounded corners, right=1.5cm of slot2, align=center] (slot3) {\textbf{Slot $n+2$}\\Producer $P_{n+2}$\\[0.3em]Executor $E_{n+1}$\\derive $s_{n+1}$; execute $\mathcal{T}_{n+1}$};

    \draw[->, thick] (slot1) -- node[above] {$\mathcal{T}_n$ fixed} (slot2);
    \draw[->, thick] (slot2) -- node[above] {$\mathcal{T}_{n+1}$ fixed} (slot3);
\end{tikzpicture}
\caption{Two-stage pipeline for VER mitigation. Block content is committed in one slot and executed in the next using a seed derived by the subsequent producer.}
\label{fig:pipeline}
\end{figure}

\subsection{Abstract specification}
Let $C_n$ be the consensus commitment to block-$n$ content, including $\mathcal{T}_n$ and a protocol-defined seed input $d_n$. Let $s_{n-1}$ be the previous seed-chain value and let $\rho_n$ denote consensus randomness required by the implementation. Producer $P_{n+1}$ evaluates a unique-output primitive:
\begin{align}
x_n &= \text{Domain} \parallel H(C_n) \parallel s_{n-1} \parallel d_n \parallel P_n \\
(y_n, \pi_n) &= \text{Eval}(sk_{P_{n+1}}, x_n) \\
s_n &= H(y_n) \\
\sigma_n &= \text{Shuffle}(\mathcal{T}_n, s_n) \\
S_{n+1} &= \text{Execute}(S_n, \sigma_n)
\end{align}

Every validating node verifies $\pi_n$, recomputes $s_n$ and $\sigma_n$, and checks the resulting state transition. The seed-chain update is committed with block $n+1$.

\begin{requirement}[Commit then randomize]
The input to the unique-output primitive must be fixed by consensus before the output that determines $\sigma_n$ becomes available to the builder of $\mathcal{T}_n$.
\end{requirement}

The original blog reflects the initial Aura implementation and describes the builder as providing a seed that the next producer signs. During the subsequent transition to BABE, the implementation adopted an sr25519/VRF path whose input combines the prior seed with BABE epoch randomness \cite{mangata2021note}. Binding the seed input to $C_n$ or its hash is a prudent defense against ambiguity, equivocation, and cross-block replay. The available materials do not establish that every revision in either implementation phase used the exact commitment binding formalized above, so it is presented here as a hardened specification requirement rather than a verified historical fact.

\subsection{Why the role split helps}
Under non-collusion, $B_n$ cannot predict the unique output of $P_{n+1}$ while selecting $\mathcal{T}_n$. It therefore cannot reliably place an attacker transaction at a chosen execution index. Conversely, $P_{n+1}$ observes the committed set before deriving or publishing the order, but it cannot add a new transaction to that already committed set.

\begin{proposition}[Conditional unilateral-order resistance]
Assume that (i) $\mathcal{T}_n$ is binding before $s_n$ is known to $B_n$, (ii) the unique-output primitive does not permit output grinding, (iii) $B_n$ and $P_{n+1}$ do not collude, and (iv) the shuffle is deterministic and correctly validated. Then neither $B_n$ nor $P_{n+1}$ can unilaterally choose an arbitrary execution permutation of $\mathcal{T}_n$.
\end{proposition}

\begin{proof}
The builder fixes the transaction set before learning the unique output that determines order. It can influence membership but cannot select an arbitrary order conditional on the realized output. The executor learns or computes the output only after the set is fixed. It can determine the protocol-prescribed permutation but cannot alter membership without invalidating the prior commitment. Unilateral arbitrary permutation therefore requires violating at least one assumption.
\end{proof}

This proposition does not eliminate all influence. The builder still chooses the set, the executor may abort or withhold, and either actor may exploit information available before its role begins. The protocol converts direct order choice into a constrained, potentially punishable availability decision.

\subsection{Seed uniqueness, non-malleability, and withholding}
The seed primitive must give one valid output for a given key and message, or otherwise prevent the producer from sampling several valid outputs and choosing the most profitable. A deterministic signature was the mechanism described in the original article. The linked prototype note uses an sr25519 key pair with VRF verification. VRFs provide public verification of a pseudorandom output and proof \cite{micali1999vrf}.

Uniqueness does not prevent an executor from computing the only valid output and then refusing to publish it when the outcome is unfavorable. Consensus must treat refusal as a missed block, invalid block, slashable omission, or other accountable failure. This is a liveness and incentive problem, not one solved by deterministic cryptography alone.

The seed chain also needs domain separation. Inputs should include chain identity, protocol version, role, block commitment, and slot or epoch information. Without domain separation, the same signed or VRF-evaluated message could be reused across contexts.

\subsection{Dependency-preserving randomization}
Let $\text{Adm}(\mathcal{T}_n)$ be the set of permutations that preserve the nonce order of every sender:
\begin{equation}
\text{Adm}(\mathcal{T}_n) = \{ \sigma : a(t_i) = a(t_j) \land q(t_i) < q(t_j) \Rightarrow t_i <_{\sigma} t_j \}.
\end{equation}

The protocol requires $\text{Shuffle}(\mathcal{T}_n, s_n) \in \text{Adm}(\mathcal{T}_n)$. Thus, Alice's transaction 1 remains before Alice's transaction 2, while Alice's and Bob's queues may be interleaved.

The public engineering note reports Fisher-Yates shuffling with Xoshiro256++ as the pseudorandom generator \cite{blackman2021scrambled, durstenfeld1964algorithm, mangata2021note}. Fisher-Yates is unbiased over an unconstrained list when its random-index sampling is unbiased. Preserving account order introduces a constrained-permutation problem. The exact distribution over $\text{Adm}(\mathcal{T}_n)$ depends on whether the implementation shuffles groups, repairs a full shuffle, or samples queue heads directly. This detail is security-critical because a biased interleaving may create statistical advantages. It should be specified and tested rather than inferred from the phrase "Fisher-Yates shuffle."

One explicit sampler is to partition transactions into nonce-ordered sender queues and repeatedly choose one active queue head using seed-derived randomness. If a queue is selected, its head is appended to the output and the next item in that queue becomes eligible. This preserves nonce order, but uniform choice among active queues is not uniform over all linear extensions. A deployment should state which fairness object it intends to approximate and measure any resulting positional bias.

\subsection{Probabilistic extraction}
Randomized ordering does not make extraction impossible. A participant can submit multiple candidate transactions to increase the chance that at least one receives a favorable position. The original proposal calls this \emph{probabilistic extraction} and argues that equal access to the lottery is fairer than exclusive builder discretion.

Per-account nonce preservation limits one simple strategy. If an account submits many sequential transactions, later transactions cannot leapfrog its first transaction. In a victim-relative race, only the account's current head may compete for the earliest position. However, the participant can create multiple accounts, so Sybil identities restore multiple independent queue heads.

The defensible guarantee is therefore conditional symmetry: if the admissible-order sampler is unbiased for its stated fairness objective, inclusion is symmetric, accounts have comparable costs, and no participant has extra information, eligible queue heads receive equal ordering probability. This is not identity-level equality and not proof of equal economic outcomes.

\subsection{Latency and client semantics}
Delayed execution changes more than consensus internals. A client submitting into block $n$ cannot treat inclusion in $n$ as final execution. Receipts, events, balance effects, pool reserves, nonce progression, and failure status become visible only after the next execution stage. Indexers and JavaScript APIs may need to fetch both block $n$ and block $n+1$ to reconstruct the transaction's effective order and result. Wallet interfaces should distinguish committed, ordered, and executed states.

Forks add a further requirement. The seed and execution result must be tied to the canonical predecessor commitment. A reorganization must not permit a seed proof or execution order from one fork to validate against a different transaction set.

\section{Mitigating Selective Denial with Intent Hiding}

\subsection{Targeted two-layer encryption}
To make a builder unable to identify which transaction is profitable to censor, the user targets a known pair consisting of builder $B_n$ and executor $E_n = P_{n+1}$. Abstractly, the wallet computes
\begin{align}
c_E &= \text{Enc}(pk_{E_n}, t \parallel m_E), \\
c_B &= \text{Enc}(pk_{B_n}, c_E \parallel m_B),
\end{align}
where $m_E$ and $m_B$ bind the ciphertext to chain, session, target slots, fee or resource limit, expiry, and replay domain. In practice, $\text{Enc}$ should be hybrid authenticated encryption: a fresh symmetric key protects the payload, and the role's public key encapsulates that symmetric key.

The workflow is shown in Figure \ref{fig:intent-hiding}.

\begin{figure}[htbp]
\centering
\begin{tikzpicture}[node distance=2cm, auto, >=latex', every node/.style={font=\small}]
    \node [draw, rectangle, rounded corners, align=center] (wallet) {User wallet\\encrypt for\\$B_n$ and $E_n$};
    \node [draw, rectangle, rounded corners, right=1.2cm of wallet, align=center] (builder) {Builder $B_n$\\remove\\outer layer};
    \node [draw, rectangle, rounded corners, right=1.2cm of builder, align=center] (block) {Block $n$\\commit\\opaque $c_E$};
    \node [draw, rectangle, rounded corners, right=1.2cm of block, align=center] (executor) {Executor $E_n$\\open, vali-\\date, execute};

    \draw[->, thick] (wallet) -- (builder);
    \draw[->, thick] (builder) -- (block);
    \draw[->, thick] (block) -- (executor);
\end{tikzpicture}
\caption{Two-layer intent-hiding path for selective VED mitigation. Encryption reduces the builder's ability to identify a profitable target, but does not make inclusion unavoidable.}
\label{fig:intent-hiding}
\end{figure}
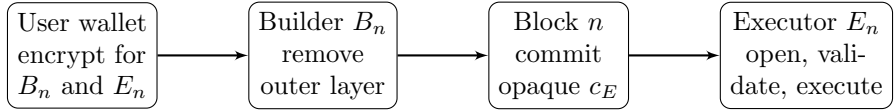

The user sends $c_B$ through the mempool or a compatible relay. The builder removes the outer layer, obtains $c_E$, and includes that opaque inner ciphertext in $\mathcal{T}_n$. The executor removes the inner layer after the set is committed, validates the plaintext transaction, and either executes it or records a deterministic failure under the chain's transaction rules.

The original article says the executor is "forced" to decrypt. More precisely, the protocol can make non-decryption detectable and punishable. Cryptography cannot force an offline or malicious validator to reveal a key-dependent plaintext. Consensus and slashing must supply accountability.

\subsection{Why submission is not node-agnostic}
Two-layer role encryption requires the wallet to know the relevant future builder and executor public keys. The ciphertext is therefore targeted, not a universal mempool object decryptable by whichever validator happens to produce the block. This requirement is easiest to satisfy in a PoS or authority schedule where future producers are predictable.

The original proposal, written against the initial Aura implementation, notes that an Aura-style schedule can let a client target a future execution slot. Aura assigns authors to time slots from a known authority set. The design can be adapted to other schedules only if equivalent role discovery, key availability, expiry, and re-targeting rules are provided. The original assertion that PoW support requires only reasonable modification remains a design hypothesis, not a complete PoW protocol.

\subsection{What each role learns}
A relay sees the outer ciphertext and associated network metadata. The builder sees the inner ciphertext after removing its layer. The executor eventually learns the plaintext. Consensus validators may learn the plaintext during validation, depending on whether decryption is published, reproduced, or embedded in execution. The protected path therefore delays disclosure; it does not necessarily provide lasting transaction privacy.

Metadata can reveal intent even when payloads are encrypted. Ciphertext size, sender identity, timing, target slot, fee class, network path, and whether a user selected the protected path may all be informative. Padding, batching, private transport, and common fee classes may be needed to make the "does not know what to deny" claim meaningful.

\subsection{Optional use and fixed resource charges}
Encryption is optional. It should be used when intent exposure creates material risk and the expected protection justifies additional cryptographic and execution cost. An encrypted payload may decode to an invalid, expired, underfunded, or otherwise unexecutable transaction. Since the builder cannot inspect it in advance, encrypted transactions must carry an enforceable fixed processing charge or prepaid resource bound.

The original DEX design used fixed transaction fees rather than a priority gas market. This does not mean execution is free. It means the cost is protocol-defined rather than competitively bid for ordering. Failed or canceled exchange transactions still pay the applicable fee or commission. The article's statement that fee-estimation algorithms are unnecessary should be read narrowly: priority-bid estimation is unnecessary, while wallets still need balance and resource estimation.

\subsection{Decryption accountability and liveness}
Once an encrypted transaction has been included for its designated executor, failure to publish the required decryption or execution result should create objective evidence. A complete implementation needs to define:
\begin{itemize}
\item the exact message and deadline the role must publish;
\item how validators distinguish malformed ciphertext from malicious non-decryption;
\item what data constitutes a slashable proof;
\item how the chain proceeds if the designated role is offline;
\item how keys and ciphertexts expire at session boundaries;
\item whether the user proves well-formed encryption before inclusion without revealing content;
\item how fees are charged when decryption or plaintext validation fails.
\end{itemize}

The original proposal suggests that a non-designated builder may still include a targeted ciphertext, thereby placing an obligation on the designated builder-executor pair when they are next able to process it. If a designated builder does not remove its layer, the ciphertext can remain in the mempool until that builder is elected again. This creates a new denial and liveness surface. To avoid undecryptable legacy ciphertexts after validator-set or key changes, encrypted transactions must not cross session boundaries unless the protocol supplies explicit re-encryption or recovery.

This mechanism does not meet level-three denial resistance. A builder can deny all encrypted transactions without knowing which one is economically valuable. It can also admit only one opaque item, hoping that its own plaintext transaction captures an opportunity.

\subsection{Role repetition and key lifecycle}
If the same validator is selected as both builder and executor for a target pair, the intended information split collapses. The protocol should reject such targeting, assign a different decryption role, or state that those slots do not provide selective-VED protection. Similar care is needed when one operator controls several validator identities.

Role keys should be distinct from consensus signing keys where operational security warrants separation. Public keys need authenticated publication, versioning, activation height, expiry, and revocation rules. A client that encrypts to a stale or wrong key must receive a deterministic outcome rather than leaving an immortal ciphertext in the network.

\section{Security Analysis}

\subsection{Conditional properties}
Table \ref{tab:security} separates intended properties from assumptions and residual failures.

\begin{table}[htbp]
\centering
\small
\begin{tabular}{p{3.5cm} p{5.5cm} p{5.5cm}}
\toprule
\textbf{Property} & \textbf{Required assumptions} & \textbf{Residual risk} \\
\midrule
Post-commit order opacity & Binding transaction set; unpredictable unique output; validated deterministic shuffle & Abort, withholding, biased constrained shuffle, collusion \\
\addlinespace
Content/order separation & Different non-colluding roles for set selection and seed derivation & Builder controls inclusion; executor controls availability \\
\addlinespace
Selective-censorship reduction & Confidential authenticated encryption; intent not leaked by metadata & Total or random denial; traffic analysis; malformed ciphertext \\
\addlinespace
Per-sender dependency preservation & Correct nonce validation and admissible-order construction & Statistical advantage across Sybil accounts; constrained-order bias \\
\addlinespace
Decryption accountability & Public deadlines, evidence rules, and credible slashing & Offline roles, weak penalties, session or key expiry \\
\addlinespace
Consensus safety & Adversary below the underlying fault threshold & Coalition at or above the threshold can suppress or rewrite evidence \\
\bottomrule
\end{tabular}
\caption{Conditional properties and residual risks of TCE-v1.}
\label{tab:security}
\end{table}

\subsection{Builder strategies}
A builder cannot choose the exact post-commit order under the assumptions of Proposition 1, but it can shape the transaction set. It may omit a victim, omit competitors, include several attacker accounts, include transactions whose validity depends on likely order, or withhold the entire block. It can also choose among mempool views or timing cutoffs if those choices are not protocol-defined. The commitment therefore removes one dimension of discretion without removing inclusion power.

The strongest defensible claim is that a non-colluding builder cannot condition an exact ordering choice on the next producer's unpredictable unique output. It can still choose a set whose many possible orders are favorable in expectation.

\subsection{Executor strategies}
An executor cannot add to the committed set, but it may learn the realized permutation before publishing a block. If the result is unfavorable, it can withhold, miss its slot, or attempt to create an invalid seed. Uniqueness makes output substitution detectable but does not remove the abort option. Penalties must exceed the expected gain from selective withholding, taking into account repeated-game effects and the probability of successful evidence inclusion.

An executor may also exploit transactions already present in its own newly built block $n+1$ based on the state resulting from executing $\mathcal{T}_n$. The pipeline should specify whether the new block's membership is chosen before, during, or after the prior set is executed, and what state information is available at each step. Otherwise, MEV can migrate across the stage boundary.

\subsection{Seed grinding and constrained choice}
Even a unique-output primitive can be undermined if the producer controls an input with several valid values. Examples include choosing among alternative block headers, timestamps, auxiliary messages, parent candidates, or transaction-set encodings. The seed message must therefore derive from canonical, consensus-fixed inputs. Any producer-controlled field included before commitment creates a grinding surface.

Repeated eligible producers create another concern. If one operator expects to control several future slots, it may condition current inclusion on seeds it can compute for its own future keys. Unpredictability must be evaluated against the actual role schedule and operator-level key ownership, not only nominal validator identities.

\subsection{Total-denial attack}
A malicious builder can exclude every user transaction except one attacker transaction. If the captured opportunity exceeds the foregone fees and expected penalty, the attack is profitable. Intent hiding does not prevent it because the builder need not identify a target.

The original article observes that if an arbitrage remains available after a builder declines it, the next builder may have a fair claim to capture it. This can reduce the first builder's incentive when opportunities persist, but many opportunities are transient or can be consumed by the attacker's own included transaction. Stronger inclusion guarantees require mechanisms outside v1, such as inclusion lists, multi-proposer construction, encrypted broadcast with availability proofs, or explicit censorship penalties.

\subsection{Builder-executor collusion}
If $B_n$ and $E_n$ cooperate, they can combine content control, seed knowledge, and both decryption layers. The pair may reveal transaction intent before execution, selectively exclude or copy a transaction, and coordinate around the derived order. Version 1 explicitly assumes such cooperation does not occur.

In 2024, after the original October 2021 blog publication, Alnajjar et al., a separate research group with no role in the original Mangata/Themis work, analyzed the two-tiered architecture and proposed publicly verifiable decryption as an extension \cite{alnajjar2024mitigating}. This third-party extension is discussed only as subsequent related work; it is not evidence of the original protocol's design or implementation. Verifiable decryption can make an incorrect opening detectable, but it does not by itself stop colluding roles from learning plaintext early or choosing not to include it. Protecting against adjacent-role collusion requires a different trust split, threshold or committee decryption, delayed-reveal cryptography, role randomization, or detection and punishment tied to observable behavior.

\subsection{Consensus-threshold coalitions}
The original article says no Byzantine-tolerant protocol can stop or punish collusion by a malicious majority. The exact boundary depends on the underlying protocol. A coalition at or above the safety, liveness, or economic threshold may censor evidence, finalize invalid execution, rewrite transaction commitments, or exempt itself from penalties. Governance capture may lower the practical threshold if the coalition can change slashing or validity rules.

The protocol should therefore state its guarantees as conditional on the same fault assumptions as consensus, plus the stronger non-collusion assumption between adjacent roles. The latter may fail even when the consensus threshold remains honest.

\subsection{Privacy boundaries}
Two-layer encryption hides transaction content from relays and the builder until the designated opening stage. It does not automatically hide sender identity, ciphertext size, timing, target slot, fee class, network path, or whether a user chose the protected path. Private mempool research emphasizes that hiding order flow changes information, incentives, price discovery, and collusion opportunities \cite{rondelet2023private}.

Because wallet behavior is part of the security boundary, light clients and wallets must implement target-role discovery, authenticated encryption, expiry, fee commitment, replay protection, and fallback behavior. A malicious or compromised wallet can reveal intent before encryption, target a colluding pair, or omit domain-binding metadata.

\section{Economic and Application Analysis}

\subsection{Spam and Sybil economics}
Random ordering creates an incentive to buy more lottery tickets by sending more transactions or using more accounts. The application-specific DEX setting addresses this economically rather than cryptographically:
\begin{itemize}
\item every included transaction pays a fixed fee, including failed transactions;
\item exchange attempts may pay a percentage commission, so repeated attempts reduce net extraction profit;
\item no priority gas market is needed to purchase an earlier position;
\item the maximum rational spam volume is bounded by the expected gain net of fees, commissions, capital costs, and penalties.
\end{itemize}

The original article argues that even one additional exchange transaction materially reduces the set of profitable extraction opportunities. This is plausible but parameter-dependent. A refereeable deployment claim requires measurements using the chain's fee schedule, trade sizes, liquidity, expected arbitrage distribution, and validator incentives.

For non-exchange transactions, the original proposal treats the ability to improve one's chance by paying for more attempts as fair because the opportunity is open to everyone. That is a normative equal-access claim. It does not account for unequal capital, account-creation costs, latency, information, infrastructure, or the ability to operate many validators.

\subsection{Dependent transactions}
Preserving sender nonce order supports multi-step workflows whose later calls depend on earlier calls. It also creates a known statistical attack: an actor can place many accounts or carefully structured queues into the batch and exploit the fact that only certain linear extensions are valid.

The original proposed hardening groups transactions by sender and shuffles subsets with the same index within each sender group. Conceptually, all first transactions compete in one stratum, all second transactions in another, and so forth. This prevents later transactions in a long single-account chain from purchasing additional early positions. It does not prevent multiple-account Sybil strategies, and it requires a precise rule for missing indices and cross-stratum dependencies.

A formal analysis should compare at least three distributions: a uniform permutation followed by rejection of nonce-invalid orders, a uniform sample from all admissible linear extensions, and a queue-head sampler. These distributions need not give the same probability to a sender's early positions.

\subsection{Price effects}
The original proposal argues that MEV redistribution commonly leaves an unsophisticated user at the end of the value chain, while random ordering yields statistically equitable prices in a sufficiently large set of trades. Randomization can remove a deterministic builder-selected disadvantage, but it cannot guarantee that every user receives a better price.

A defensible statement is conditional: if a batch is fixed, the order sampler is symmetric, no participant has privileged information, and price impact is the only relevant outcome, then no eligible transaction is assigned a systematically preferred index by the builder. Realized execution prices still depend on sampled order, trade size, slippage limits, liquidity, and other participants. A transaction may receive a worse price than under another protocol realization even when the expected allocation is less builder-biased.

\subsection{Conventional incentives remain necessary}
The protocol does not remove block rewards, transaction fees, or commissions. These are conventional MEV under the terminology adopted by the original article and are treated as necessary incentives. Their design matters because validators must be compensated for delayed execution, seed proofs, ciphertext handling, and state reconstruction.

Fixed fees remove the explicit bidding channel used by priority gas auctions, but they can also underprice expensive encrypted failures or overprice simple calls. A production design should meter worst-case decryption, validation, and failure handling while preventing data-dependent fees from leaking transaction intent.

\subsection{Application-specific completeness}
No consensus extension is automatically complete for every application. A DEX can expose cross-pool arbitrage, liquidation, oracle-update ordering, liquidity-addition timing, governance interactions, and cross-chain state dependencies. Other applications expose different surfaces. The original article correctly requires each implementation to enumerate its MEV opportunities and test whether the introduced measures are sufficient.

The VER/VED decomposition is useful for this audit. For every opportunity, designers should ask who can learn the relevant intent, who can change membership, who can influence relative order, what information arrives between commitment and execution, and whether an adjacent or cross-domain transaction can recreate the opportunity.

\section{Implementation Report}

\subsection{Public Substrate design note}
A later engineering note documents the BABE implementation of the Rust/Substrate design, in which execution of transactions included in block $N$ is delayed until block $N+1$ \cite{mangata2021note}. It records the following components:
\begin{itemize}
\item a modified block-building path that executes extrinsics from the previous block during creation of the current block;
\item randomized shuffling that preserves the original per-account order because later calls from one account may depend on earlier calls;
\item Fisher-Yates shuffling driven by Xoshiro256++ and initialized from a seed stored in blockchain state;
\item an sr25519 key pair and VRF verification so that other nodes can validate the producer's seed using the block author's public key;
\item a seed input composed from the prior seed and BABE epoch randomness;
\item seed injection as Substrate \texttt{InherentData}, with verification during block import;
\item a dedicated runtime \texttt{RandomSeedApi} for public seed access;
\item changes involving \texttt{sc-block-builder}, \texttt{sc-basic-authorship}, \texttt{sc-consensus-babe}, and \texttt{sc-service}.
\end{itemize}

The note gives a concrete ordering example. An original list containing Alice's first, second, and third transactions followed by Bob's first and second can become an interleaving such as Alice 1, Bob 1, Alice 2, Bob 2, Alice 3. The relative order within each account is preserved.

The same note states that a producer cannot calculate the seed for a block it is about to author until the preceding block exists, and that validators reject a block whose seed fails VRF verification. These statements support public verifiability and delayed availability. They do not by themselves prove absence of all grinding inputs or selective withholding.

\subsection{Implementation evolution from Aura to BABE}
The October 2021 article was written against the initial implementation, which used Aura consensus and was intended to connect as a Polkadot parachain \cite{mangata2021themis}. After that publication, the Mangata team transitioned the implementation over time to BABE. The linked engineering note documents this later BABE phase, including epoch randomness, inherent-data injection, and seed verification \cite{mangata2021note}. The two sources therefore describe successive implementation stages rather than conflicting accounts.

The protocol itself treats the consensus engine as an integration parameter. It needs a known or discoverable future role, a public verification key, canonical seed inputs, and a validity rule for the seed proof. For reproducibility, the final submission should cite immutable commits or releases for both the initial Aura phase and the later BABE phase, and identify which revision supports each implementation claim.

\subsection{Client and event reconstruction}
Delayed execution changes the relationship between block inclusion and execution events. The engineering materials describe JavaScript API work that retrieves the block containing the transactions and the following block containing their execution context, then reconstructs the shuffled order and associated events. This is not merely a user-interface detail. Explorers, wallets, accounting systems, and downstream protocols need a canonical mapping from a committed transaction to its later result.

A reproducible implementation should expose:
\begin{itemize}
\item the committed transaction-set identifier;
\item the seed and proof used for its order;
\item the deterministic ordered list;
\item success, failure, fee, and commission events in that order;
\item the state root before and after execution;
\item reorganization behavior for both commitment and execution blocks.
\end{itemize}

\subsection{Implementation status and reproducibility}
The original article reported that Themis had been fully implemented by the Mangata DEX blockchain in Substrate/Rust and linked a public repository \cite{mangata2021repo, mangata2021themis}. The repository has since moved or been renamed, so a durable paper should cite a tagged release or immutable commit rather than only a moving branch.

The public engineering note documents the delayed-execution and shuffle path in useful detail. The original article stated that a more technical VED specification would be published later. Consequently, the public sources support a stronger historical implementation claim for VER than for the complete two-layer VED workflow.

A technical presentation at Parity's Sub0 conference discussed the Substrate block-execution changes, with a later Mangata recap naming Mateusz Nowakowski as the presenter \cite{mangata2021techtalk}. For reproducibility, an arXiv release should ideally include:
\begin{itemize}
\item an immutable source commit and build instructions;
\item protocol test vectors for seed derivation and constrained shuffling;
\item adversarial tests for seed withholding, malformed ciphertext, and missed decryption;
\item benchmarks for throughput, latency, cryptographic cost, and event reconstruction;
\item economic simulations for spam, multi-account extraction, and fee calibration;
\item a statement of which VED components were implemented, tested, or only proposed.
\end{itemize}

\section{Related Work}

\subsection{MEV measurement and redistribution}
Daian et al.\ introduced the term \emph{miner extractable value} in their analysis of DEX arbitrage, priority gas auctions, transaction-ordering dependencies, and consensus instability \cite{daian2020flash}. Flashbots subsequently proposed transparent and permissionless infrastructure to illuminate extraction, democratize participation, and distribute benefits \cite{obadia2020frontrunning}. Such systems can reduce inefficient public bidding and expose MEV flows, but they do not necessarily remove the underlying ordering opportunity. TCE-v1 instead seeks to minimize unilateral extraction power at the execution layer.

\subsection{Order-fair consensus}
Aequitas formalizes transaction order-fairness and provides consensus constructions that augment consistency and liveness with ordering properties \cite{kelkar2020order}. Wendy provides modular fairness mechanisms and also shows that some intuitive fairness definitions are impossible \cite{kursawe2020wendy}. These works relate network observations or collective preferences to final order. TCE-v1 intentionally avoids receive- or send-order claims and uses a post-commit random order.

Kelkar et al.\ published a separate protocol named Themis that targets strong order-fair Byzantine consensus \cite{kelkar2023themis}. It is unrelated to the Mangata Themis Consensus Extension. The shared name makes the full term "Themis Consensus Extension" or the abbreviation TCE-v1 important in citations and metadata.

\subsection{Encrypted mempools and verifiable decryption}
Private or encrypted mempools hide transaction content until a commitment or ordering point. Their security depends not only on cryptography but also on decryption incentives, information flows, validator concentration, and collusion \cite{rondelet2023private}. The VED component is a targeted two-role encrypted-mempool design, with lower coordination complexity than a threshold committee but a stronger dependence on the honesty and availability of adjacent roles.

In 2024, Alnajjar et al., a separate research team not involved in the original 2021 Mangata/Themis work, studied the two-tiered architecture and proposed a verifiable-decryption extension \cite{alnajjar2024mitigating}. Their paper is third-party, subsequent related work. Its mechanism is not part of the 2021 v1 protocol or its original implementation, and it should not be read as evidence about the original team's design. It nevertheless highlights two open points: decryption must be publicly accountable, and collusion between adjacent roles requires explicit treatment.

\section{Known Attacks and Future Work}
The original proposal identifies three immediate attack classes:
\begin{description}
    \item[Statistical attack on non-exchange transactions.] Per-sender nonce preservation creates structured admissible orders. A participant may spend more on fees and accounts to increase the chance of a favorable position. The proposed next step is to group by sender and shuffle equal-index strata. A complete solution needs a formal sampler, a Sybil-cost model, and statistical tests for bias.
    \item[Total denial except one transaction.] A builder may exclude all user traffic and include only its own transaction when the gain exceeds foregone fees and penalties. Future work should combine intent hiding with enforceable inclusion or multi-party block construction.
    \item[Collusion.] Version 1 assumes malicious actors act independently. Future versions should detect or punish cooperation between builders, executors, relays, and searchers. Observable evidence is difficult because private cooperation may be behaviorally indistinguishable from ordinary inclusion choices.
\end{description}

The original roadmap also calls for better management of encrypted transactions, mitigation of encrypted-transaction failure, detection of cooperating malicious validators, and a direct response to builder-executor collusion. Additional research needed for an implementable v2 includes:
\begin{itemize}
\item session-safe key lifecycle, re-targeting, expiry, and recovery;
\item formal decryption evidence and penalty calibration;
\item unbiased or explicitly characterized sampling over nonce-constrained permutations;
\item protection against metadata and timing leakage;
\item PoW or unpredictable-proposer adaptations;
\item inclusion guarantees that address total denial;
\item empirical evaluation of latency, throughput, market quality, and validator revenue;
\item analysis of cross-domain and cross-block MEV not captured by one application batch;
\item operator-level collusion analysis when one entity controls several validator keys.
\end{itemize}

\section{Conclusion}
Themis Consensus Extension v1 separates block-content commitment from execution-order derivation and adds an optional two-layer intent-hiding path. For an application-specific chain with predictable producer roles, fixed transaction costs, accountable validators, and non-colluding adjacent producers, the design reduces two important unilateral powers: choosing an exact order after selecting content and selectively censoring a transaction after reading its purpose.

The resulting guarantee is best described as fairer execution opportunity under randomized ordering, not order fairness in the send- or receive-time sense. It can remove classes of deterministic frontrunning and selective denial under its assumptions, but it does not eliminate probabilistic extraction, arbitrary or total denial, role withholding, or collusion. Randomized order can improve expected treatment relative to a builder-chosen worst position, yet it does not guarantee a better realized price for every trade.

The design is most directly suited to PoS or authority systems in which future roles and public keys are known. Wallets and light clients are integral because they must target encryption to those roles. Each application-specific deployment must examine its own MEV surface, fee economics, and state dependencies. Finally, no extension can preserve fairness when an adversarial coalition exceeds the underlying consensus or governance fault threshold.

\section*{Acknowledgments}
The original public article states that the research was led by Gleb Urvanov and thanks Xinshu Dong, Luke Pearson, Will Wolf, Peter Kris, Marcin Gorny, and the Substrate Builders team for comments and ideas \cite{mangata2021themis}.

\appendix
\section{Protocol Questions Preserved from the Original Proposal}
This appendix collects the operational questions raised in the 2021 article and points to their academic treatment.

\begin{table}[htbp]
\centering
\small
\begin{tabular}{p{5.5cm} p{9cm}}
\toprule
\textbf{Original question} & \textbf{Answer in this specification} \\
\midrule
Is the protocol receive-order or send-order fair? & No. It uses post-commit randomized ordering; see Sections 2.3 and 5. \\
\addlinespace
How is probabilistic extraction handled? & Per-account nonce order limits repeated transactions from one account, while fees and commissions price attempts. Multiple accounts remain possible; see Sections 5.6 and 8.1. \\
\addlinespace
Does randomization encourage spam? & Yes. Fixed fees, failed-transaction charges, and commissions create an economic cap but not a cryptographic prohibition; see Section 8.1. \\
\addlinespace
How are dependent transactions handled? & Sender nonce order is preserved while accounts are interleaved; see Sections 5.5 and 8.2. \\
\addlinespace
How is future decryption guaranteed? & It is not cryptographically guaranteed. Inclusion creates an obligation whose violation must be provable and punishable; see Section 6.5. \\
\addlinespace
What if the designated pair does not decrypt? & The ciphertext may wait for the designated role, creating liveness attacks. It must expire before incompatible session or key changes; see Section 6.5. \\
\addlinespace
Can another builder include a targeted ciphertext? & The original design permits this, but the designated roles may still be needed later. The resulting obligation, storage, and expiry rules require specification; see Section 6.5. \\
\addlinespace
Can a user select execution time? & In predictable Aura-style schedules, a wallet can target a role pair or slot. This does not generalize automatically to all PoS or PoW systems; see Section 6.2. \\
\addlinespace
Does random ordering guarantee a better trade price? & No per-trade guarantee follows. It can remove builder-selected systematic disadvantage under symmetric assumptions; see Section 8.3. \\
\addlinespace
What attacks remain? & Nonce-structure statistics, total denial, role withholding, Sybil accounts, metadata leakage, and collusion; see Sections 7, 8, and 11. \\
\bottomrule
\end{tabular}
\caption{Original protocol questions and their formal treatment.}
\end{table}


\end{document}